\documentclass[reqno]{amsart}
\usepackage{amssymb,epsfig,graphics,graphicx,bbm,hyperref,color}
\usepackage{xfrac,mathtools,MnSymbol}
\usepackage{epstopdf}
\usepackage{multicol}
\usepackage{enumitem}

\definecolor{gray}{rgb}{0.93,0.93,0.93}
\definecolor{light-gold}{rgb}{0.99,0.97,0.78}

\def\be{\begin{equation}}
\def\ee{\end{equation}}
\def\bm{\begin{multline}}
\def\eem{\end{multline}}
\def\bfig{\begin{figure}[htb]}
\def\efig{\end{figure}}

\newcommand{\e}[1]{\,{\rm e}^{#1}\,}

\def\tr{{\operatorname{tr\,}}}

\numberwithin{equation}{section}
\newtheorem{theorem}{Theorem}[section]
\newtheorem{proposition}[theorem]{Proposition}

\newtheorem{remark}[theorem]{Remark}

\newcommand{\caE}{{\mathcal E}}
\newcommand{\caH}{{\mathcal H}}

\newcommand{\bbC}{{\mathbb C}}

\newcommand{\bbZ}{{\mathbb Z}}

\newcommand{\eps}{{\varepsilon}}

\newcommand{\EE}{\mathbb{E}}

\newcommand{\PP}{\mathbb{P}}

\newcommand{\RR}{\mathbb{R}}

\newcommand{\cE}{\mathcal{E}}

\renewcommand{\a}{\alpha}
\renewcommand{\b}{\beta}
 
\renewcommand{\d}{\delta}

\renewcommand{\L}{\Lambda}
\newcommand{\om}{\omega}

\renewcommand{\S}{\Sigma} 
\newcommand{\s}{\sigma}

\newcommand{\es}{\varnothing}
\newcommand{\se}{\subseteq}
\newcommand{\sm}{\setminus}
\newcommand{\lra}{\leftrightarrow}
\newcommand{\crit}{\mathrm{c}}
\newcommand{\oo}{\infty}
\newcommand{\one}{\hbox{\rm 1\kern-.27em I}}

\makeatletter
  \def\tagform@#1{\maketag@@@{\footnotesize{(#1)}\@@italiccorr}}
\makeatother

\renewcommand{\eqref}[1]{(\ref{#1})}

\begin{document}

\title[Heisenberg models on trees, via   random loops]
{Critical temperature of Heisenberg models on regular trees, via
  random loops} 

\author{Jakob E. Bj\"ornberg}
\address{
Department of Mathematical Sciences,
Chalmers University of Technology and the University of Gothenburg,
Sweden}
\email{jakob.bjornberg@gu.se}

\author{Daniel Ueltschi}
\address{Department of Mathematics, University of Warwick,
Coventry, CV4 7AL, United Kingdom}
\email{daniel@ueltschi.org}



\begin{abstract}
We estimate the critical temperature of a family of quantum 
spin systems on regular trees of large degree.
The systems include the spin-$\frac12$ XXZ model and the spin-1 nematic model. Our formula is conjectured to be valid for large-dimensional cubic lattices. Our method of proof uses  a probabilistic
representation in terms of random loops.
\end{abstract}

\thanks{\copyright{}2018 by the authors. This paper may be reproduced, in its
entirety, for non-commercial purposes.}

\maketitle

\section{Introduction and main result}

The main goal of this study is to predict an expression for the
critical temperature of a family of quantum spin systems on the cubic
lattice $\bbZ^\nu$ that holds asymptotically for large dimension
$\nu$. More precisely, we propose the first two terms in the expansion
in powers of $\nu^{-1}$. The family of quantum spin systems includes
the spin $\frac12$ ferromagnetic and antiferromagnetic Heisenberg
models and the XXZ model. We also consider spin 1 quantum nematic
systems. Our results are expected to be exact but they are not
rigorous on $\bbZ^\nu$. In fact we do not perform calculations with
the cubic lattice but we consider the model on regular trees
with $d$ descendants; we obtain the first two terms of the critical
inverse temperature in powers of $d^{-1}$. For trees our computations
are completely rigorous.  We conjecture that our
expression  applies to $\bbZ^\nu$ when taking $d = 2\nu-1$.

\subsection{Random-loop model}
\label{sec random loop model}

Our method is based on using a random loop representation, which we
now describe.  
The relevant model of random loops may be 
defined for arbitrary finite graphs, here we consider mainly trees.
Let $T$ denote an infinite rooted tree 
where each vertex has $d\geq2$
offspring, and write $\rho$ for its root.
We sometimes refer to the number of offspring of vertex as its
\emph{outdegree}. 
For $m\geq 0$ let $T_m$ denote the subtree of $T$ consisting of
the first $m$ generations ($\rho$ being generation zero).
Write $V_m$ and $E_m$ for the vertex- and edge sets of $T_m$.

Let $\PP_m(\cdot)$ denote a probability measure 
governing a collection $\om=(\om_{xy}:xy\in E_m)$ of
independent Poisson processes on the interval $[0,1]$, indexed by the
edge-set $E_m$, each having rate $\b$ (the inverse-temperature).
We refer to realizations of
$\om$ as a collections of \emph{links}, and to $\om_{xy}$
as the links \emph{supported by} the edge $xy$.
Thus, disjoint sub-intervals $I,J\se[0,1]$
independently receive uniformly placed links, their number being 
Poisson-distributed with mean $\b|I|$ and $\b|J|$, respectively.
We write $\EE_m[\cdot]$ for expectation under $\PP_m(\cdot)$.

A given link is assigned to be a \emph{cross} with probability $u$,
otherwise a \emph{double-bar}, independently between different links.
The collection of links then decomposes $T_m\times[0,1]$ into a
collection of disjoint \emph{loops} in a natural way.  Rather than
giving a formal definition here, we refer to Fig.~\ref{fig loops}.
A formal definition may be found e.g.\ in 
\cite[Sect. 2.1]{Uel13}.
\bfig
\includegraphics[width=45mm]{loops.1}
\caption{Random loops coming from a configuration $\om$
  of crosses and bars, in the case when the underlying graph is a line
  with seven vertices.  To each vertex corresponds a vertical line
  segment which is a copy of the interval $[0,1]$. 
On following a loop one reverses direction when traversing a
  double-bar, maintains direction when traversing a cross, and proceeds
  periodically in the vertical 
  direction.  In this example there are $\ell(\om)=4$ loops.}
\label{fig loops}
\efig

The total number of loops is denoted $\ell=\ell(\om)$. 
We actually work with a weighted version of $\PP_m(\cdot)$,  denoted
$\PP^{(\theta)}_m(\cdot)$ with a positive parameter $\theta$.
This is the probability measure whose expectation operator
$\EE_m^{(\theta)}[\cdot]$ is given by 
\[
\EE^{(\theta)}_m[X]=\frac{\EE_m[X\theta^{\ell(\om)}]}
{\EE_m[\theta^{\ell(\om)}]}.
\]
Note that $\PP^{(1)}_m=\PP_m$.

All loops are small when $\beta$ is small,
this may be shown e.g.\ as in \cite[Thm. 6.1]{GUW}.
 But it is expected that
there exists $\beta_\crit$, that depends on the parameter 
$\theta$ and the outdegree $d$, such that a given points lies 
in an infinite loop 
with positive probability for $\beta > \beta_\crit$. 
Our main result is a formula for $\beta_\crit$; it is asymptotic
in the outdegree $d\to\oo$, namely
\be
\label{main formula}
\frac{\beta_\crit}{\theta}=\frac{1}{d}+
\frac{1-\theta u(1-u)-\tfrac16\theta^2(1-u)^2}{d^2}
+o(d^{-2}),
\ee
and we can prove that there are infinite loops for $\beta>\beta_\crit$
in the vicinity of $\beta_\crit$.
For a more precise statement, see Theorem \ref{thm} below.

The first study of this model on trees is due to Angel \cite{Ang}, who
established the presence of long loops for a range of parameters
$\beta$ when $d\geq4$; 
he only considered the case $u=1$ and $\theta=1$. Angel's
results were extended by Hammond \cite{Ham1,Ham2}; he gave a precise
characterisation of the critical parameter $\beta_\crit$ for large
enough $d$.  The formula
\eqref{main formula} was established in \cite{BU-tree} in the case
$\theta=1$, our study following a suggestion of Hammond. 
Very recently Hammond and Hegde \cite{HH} proved that
the formula \eqref{main formula} for $\theta=1$ truly identifies the
critical point, not only in the local sense considered here and in
\cite{BU-tree};  their results hold for large $d$.
Another extension to $\theta \neq 1$ has independently
 been proposed by Betz,
Ehlert, and Lees \cite{BEL}.

\subsection{Quantum spin systems}
\label{sec quantum model}

Let $(\Lambda,\caE)$ denote a graph, with $\Lambda$ the set of
vertices and $\caE$ the set of edges. The
main examples to bear in mind here are
regular trees, and finite subsets of $\bbZ^d$ with nearest-neighbour
edges. The spin-$\frac12$ systems have Hilbert space 
$\caH_\Lambda = \otimes_{x\in\Lambda} \bbC^2$ 
and the hamiltonian is 
\[
H_\Lambda =
-2 \sum_{\{x,y\} \in \caE} \Bigl( S_x^{(1)} S_y^{(1)} + S_x^{(2)}
S_y^{(2)} + \Delta S_x^{(3)} S_y^{(3)} \Bigr), 
\] 
where $S_x^{(i)}$,
$i=1,2,3$ denotes the $i$th spin operator at site $x \in
\Lambda$. Here, $\Delta \in [-1,1]$ is a parameter.

As was progressively understood in \cite{Toth, AN, Uel13}, this
quantum system is represented by the model of random loops with
$\theta=2$ and $u = \frac12 (1+\Delta)$. Indeed, the quantum two-point
correlation function is given by loop correlations, 
\be\label{corr}
\langle S^{(1)}_x S^{(1)}_y\rangle:=
\frac{\tr\big(S^{(1)}_xS^{(1)}_y\e{-\b H_\Lambda}\big)}
{\tr\big(\e{-\b H_\Lambda}\big)} =\tfrac14
\PP_\Lambda^{(\theta=2)}(x\lra y), 
\ee 
where $\{x\lra y\}$ is the
event that $(x,0)$ and $(y,0)$ belong to the same loop. It follows
that magnetic long-range order is related to the occurrence of large
loops.

On $\bbZ^3$, the critical inverse temperature has been computed numerically; it was found that
\be
\label{numerical values}
\beta_{\rm c}^{(\nu=3)}(\Delta) = \begin{cases} 0.596 & \text{if $\Delta = 1$; Troyer et.al.\ \cite{TAW}} \\
0.4960 & \text{if $\Delta=0$; Wessel, private communication in \cite{BBBU}} \\
0.530 & \text{if $\Delta=-1$; Sandvik \cite{San}, Troyer et.al.\ \cite{TAW}} \end{cases}
\ee

For large $\nu$, the lattice $\bbZ^\nu$ behaves like a tree of outdegree $d=2\nu-1$. Our formula \eqref{main formula} gives
\be
\label{formula 2}
\beta_{\rm c}^{(\nu)}(\Delta) = \tfrac1\nu + \tfrac1{\nu^2} \bigl[ 1 - \tfrac16 (1-\Delta)(2+\Delta) \bigr] + o(\tfrac1{\nu^2}).
\ee
With $\nu=3$, the values for $\Delta = 1,0,-1$ are $\frac49, \frac{11}{27}, \frac{11}{27}$ respectively. They corroborate the numerical values \eqref{numerical values} to some extent. Of course, the formula \eqref{formula 2} gets more accurate in high dimensions.

In the case of spin-1 systems, the Hilbert space is $\caH_\Lambda = \otimes_{x\in\Lambda} \bbC^3$ and the hamiltonian is
\be
\label{ham spin 1}
H_\Lambda = -\sum_{\{x,y\} \in \caE} \Bigl( u \vec S_x \cdot S_y + (\vec S_x \cdot S_y)^2 \Bigr),
\ee
See \cite{Uel13}. The phase diagram of this model was
determined in \cite{FKK}. For $0<u<1$ the system displays nematic long-range order at low temperatures (if $d\geq3$; also in the ground state when $d=2$). This was rigorously proved in \cite{TTI,Uel13}. The corresponding loop model has parameter $\theta=3$, and the same $u$ as in \eqref{ham spin 1}.
Loop correlations are related to nematic long-range order, namely
\be
\label{corr2}
\langle A_x A_y \rangle = \tfrac29  \PP_\Lambda^{(\theta=3)}(x\lra y),
\ee
with $A_x = (S_x^{(3)})^2 - \frac23$. We are not aware of numerical calculations of the critical inverse temperature $\beta_{\rm c}$ for this model on $\bbZ^3$. With $\theta=3$ and $d=2\nu-1$, the formula \eqref{main formula} gives
\[
\beta_{\rm c}^{(\nu)}(u) = \tfrac3{2\nu} + \tfrac3{2\nu^2} 
\bigl[ 1 - \tfrac34 (1-u^2) \bigr] + o(\tfrac1{\nu^2}).
\]

\subsection{Main result}

In the rest of this paper we deal only with the probabilistic model of
random loops defined above, and we allow $\theta$ to be any (fixed)
positive real number.  Our main result is that, as the distance between $x$ and
$y$ goes to $\oo$, the two-point function vanishes or stays positive,
according to whether $\b$ is smaller or larger than $\b_\crit$ given
above.
Let us say that a loop \emph{visits} a vertex $x$ of $T_m$ if the loop
contains a point $(x,t)$ for some $t\in[0,1]$.
Motivated by \eqref{corr} and \eqref{corr2} we consider
\[
\s_m = \PP^{(\theta)}_m(\rho \lra m),
\]
that is, $\s_m$ is the $\PP^{(\theta)}_m$-probability 
that $(\rho,0)$ belongs to a  loop which visits some vertex in 
generation $m$ in $T_m$.

Throughout this paper 
we work with $\beta$ of the form
\be\label{beta-form}
\frac{\beta}{\theta}=\frac1d+\frac\alpha{d^2}, \qquad
\mbox{where } |\alpha|\leq\alpha_0
\ee
for some fixed but arbitrary $\alpha_0>0$.  All error terms
$O(\cdot)$, $o(\cdot)$
and constants may depend on $\alpha_0$ but are otherwise uniform in
$\a$.  

\begin{theorem}\label{thm}
Consider $\b$ of the form \eqref{beta-form}, and write
\[
\a_\ast=\a_\ast(\theta,u)=1-\theta u(1-u)-\tfrac16\theta^2(1-u)^2.
\]
For any $\d>0$ there exists
$d_0=d_0(\theta,u,\a_0,\d)$ such that for $d\geq d_0$ we have:
\begin{itemize}
\item if $\a\leq\a_\ast-\d$ then
$\lim_{m\to\oo} \s_m=0$;
\item if $\a\geq\a_\ast+\d$ then
$\liminf_{m\to\oo} \s_m>0$.
\end{itemize}
\end{theorem}

Let us remark that for $\theta=1$ 
the result was shown in our previous
work \cite{BU-tree}.
The arguments presented here are  strengthened versions of those
arguments. 
The basic strategy is to establish recursion inequalities for the
sequence $\s_m$, see Prop.\ \ref{ineq-lem}.
These are obtained by analyzing the local configuration
around the root $\rho$, in particular we identify two events $A_1$ and
$A_2$ which together contribute
most of the probability in the regime we consider ($d\to\oo$
and $\b$ as in \eqref{beta-form}).

\section{Proof of the main result}

The indicator function of an event $A$ will be written $\one_A$
or $\one\{A\}$.  
The partition function for the loop model on $T_m$ is written
$Z_m=\EE_m[\theta^\ell]$.  For convenience we also define
\be\label{def z}
z_m=e^{-d\b(1-1/\theta)}\frac{\theta Z_{m-1}^d}{Z_m}.
\ee

For given $m\geq1$ and $\eps>0$ we define 
\[
\tilde \s_m = \s_m\wedge \s_{m-1}\wedge (\tfrac\eps d).
\]
(A priori we need not have $\s_m\leq\s_{m-1}$ since they are computed
using different measures.)
In this section we will prove the following recursion-inequalities.

\begin{proposition}\label{ineq-lem}
For all $m\geq1$ we have 
\be\label{large-ineq}
\s_m\geq \tilde \s_{m-1}+
\tfrac{\tilde \s_{m-1}}{d}\big(\alpha-\a_\ast\big)-\tfrac12\tilde\s_{m-1}^2
+ O(d^{-3}),
\ee
and
\be \label{not-large-ineq}
\s_m\leq
(\s_{m-1}\vee\s_{m-2})\big[1+\tfrac1d\big(\a-\a_\ast\big)+O(d^{-2})
\big].
\ee
Here the $O(d^{-3})$ and $O(d^{-2})$ are uniform in $m$.
\end{proposition}

Our main result follows easily:
\begin{proof}[Proof of Thm \ref{thm}]
First suppose 
$\a<\a_\ast$.  For $d$ large enough the factor in square
brackets in \eqref{not-large-ineq} is strictly smaller than 1.  This
easily gives that $\s_m$ decays to 0 exponentially fast.

Now suppose 
$\a>\a_\ast$.   Clearly
$\s_0=1$, and it is not hard to see that there exists a constant $c_1>0$
such that $\s_1\geq c_1$ for all $d$.  This implies that 
$\tilde \s_1=\eps/d$ if $\eps<c_1$.
If also $\eps<2(\a-\a_\ast)$  and $d$ is large enough
then \eqref{large-ineq} and induction on $m$ give that 
$\s_m\geq\tilde \s_{m}= \eps/d$ for all $m\geq1$.  
\end{proof}

Before turning to the proof of Prop. \ref{ineq-lem}, let us describe
some of the main ideas and also 
what new input is required compared
to our previous work \cite{BU-tree} on the case $\theta=1$.
For the lower bound \eqref{large-ineq} we will estimate the
probability of certain local configurations near $\rho$ which
guarantee that $\rho$ is connected to generation $m$ if certain of its
children (or grandchildren) are.  For the upper bound we similarly estimate
$\PP_m^{(\theta)}(\rho\not\lra m)$ in terms of the probability that
certain of $\rho$'s children (or grandchildren) are blocked from generation
$m$.  When $\theta\neq1$, the configurations in the subtrees rooted at
the children  of $\rho$ are not independent of the local configuration
adjacent to $\rho$.  Thus we must deal carefully with the factor
$\theta^{\ell(\om)}$ and how it behaves in the local configurations
which we consider.  This involves obtaining estimates for the
partition function $Z_m$ in terms of the partition function $Z_{m-1}$
in the smaller tree, which is where the number $z_m$ in \eqref{def z}
becomes relevant.

As was the case in \cite{BU-tree}, the hardest part is the upper bound
\eqref{not-large-ineq}.  This is because we must rule out connections
due to `lower order events' ($(A_1\cup A_2)^c$ in the notation
below) where the loop structure is too complicated to handle directly.  
The main technical advance compared to \cite{BU-tree} started
with a simplification of the argument used there to deal with this
difficulty.  Having this simpler version allowed us to deal also with
the correlations caused by the factor $\theta^{\ell(\om)}$, see
Prop. \ref{tree-bound}. 

\subsection{Preliminary calculations}

Let us first introduce some notations and prove some facts
that will be used for establishing both bounds in 
Prop~\ref{ineq-lem}.  

Write $A_1$ for the event that, for each child $x$ of $\rho$, there is
at most one link between $\rho$ and $x$.  
Write $A_2$ for the event that:  (i) there is a unique child $x$ of
$\rho$ with exactly 2 links between $\rho$ and $x$, (ii) for all
siblings $x'$ of $x$ there is at most one link between $\rho$ and
$x'$, and (iii) for all children $y$ of $x$ there is at most one link
between $x$ and $y$.  See Fig.~\ref{fig events}.

\begin{centering}
\bfig
\begin{picture}(0,0)%
\epsfig{file=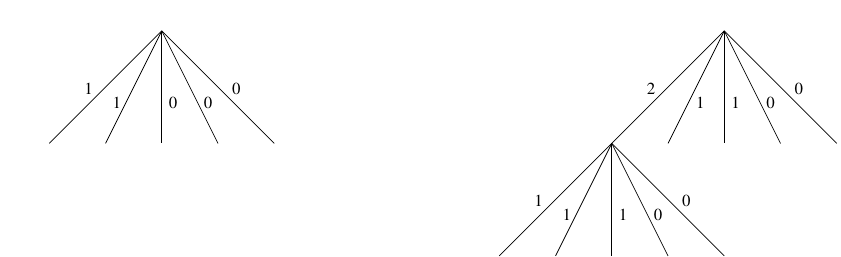}%
\end{picture}%
\setlength{\unitlength}{1776sp}%
\begingroup\makeatletter\ifx\SetFigFont\undefined%
\gdef\SetFigFont#1#2#3#4#5{%
  \reset@font\fontsize{#1}{#2pt}%
  \fontfamily{#3}\fontseries{#4}\fontshape{#5}%
  \selectfont}%
\fi\endgroup%
\begin{picture}(8937,2718)(676,-2773)
\put(5551,-586){\makebox(0,0)[lb]{\smash{{\SetFigFont{9}{10.8}{\rmdefault}{\mddefault}{\updefault}{\color[rgb]{0,0,0}$A_2$:}%
}}}}
\put(2251,-211){\makebox(0,0)[lb]{\smash{{\SetFigFont{7}{8.4}{\rmdefault}{\mddefault}{\updefault}{\color[rgb]{0,0,0}$\rho$}%
}}}}
\put(8251,-211){\makebox(0,0)[lb]{\smash{{\SetFigFont{7}{8.4}{\rmdefault}{\mddefault}{\updefault}{\color[rgb]{0,0,0}$\rho$}%
}}}}
\put(7051,-1486){\makebox(0,0)[lb]{\smash{{\SetFigFont{7}{8.4}{\rmdefault}{\mddefault}{\updefault}{\color[rgb]{0,0,0}$x$}%
}}}}
\put(676,-586){\makebox(0,0)[lb]{\smash{{\SetFigFont{9}{10.8}{\rmdefault}{\mddefault}{\updefault}{\color[rgb]{0,0,0}$A_1$:}%
}}}}
\end{picture}%
\caption{Illustrations of the two events $A_1$ and $A_2$.  Numbers on
  edges indicate the number of links.}
\label{fig events}
\efig
\end{centering}

Let $\zeta_m=1-\s_m$ and let
$B^\rho_m$ be the event that $(\rho,0)$ 
does \emph{not} belong to a
loop which reaches generation $m$ in $T_m$, thus
$\PP^{(\theta)}_m(B^\rho_m)=\zeta_m$. 
Clearly we have that
\be\label{zeta-split-eq}
\zeta_m=\PP^{(\theta)}_m(B^\rho_m)=
\PP^{(\theta)}_m(B_m^\rho\cap A_1)
+\PP^{(\theta)}_m(B_m^\rho\cap A_2)
+\PP^{(\theta)}_m(B_m^\rho\sm (A_1\cup A_2)).
\ee
Let us enumerate the children of $\rho$ by $i=1,\dotsc,d$
and let $\ell_i$ denote the number of loops in the restriction of
$\om$ to the subtree to distance $m$ rooted at child $i$.
On the event $A_1$, and if there are $k$ links from $\rho$,
the number $\ell$ of loops satisfies
\be\label{ell-A1}
\ell=\textstyle\sum_{i=1}^d \ell_i -k +1.
\ee
To see this, one may imagine that the $k$ links to $\rho$ are put in
last, one at a time.  Each such link then merges some loop in the
corresponding subtree with a loop visiting $\rho$.
(This uses the tree-structure of the underlying graph, which implies
that there can be no connections between $\rho$ and 
the subtree until the link is put in.)
It follows that
\[\begin{split}
\EE_m[\theta^\ell\one_{A_1}]&=\sum_{k=0}^d
\theta^{-k+1}\EE_m\Big[\theta^{\sum_i\ell_i}
\one_{A_1}\one\{k\mbox{ links at }\rho\}\Big]
=\theta \sum_{k=0}^d 
\binom{d}{k} (e^{-\b})^{d-k} (e^{-\b}\tfrac\b\theta)^k
Z_{m-1}^{d} 
\\
&=\theta Z_{m-1}^d\big(e^{-\b}(1+\tfrac\b\theta)\big)^d
\end{split}\]
and hence (recalling $z_m$ from \eqref{def z})
\be\label{PA1}
\PP^{(\theta)}_m(A_1)=
z_m
\big(e^{-\b/\theta}(1+\tfrac\b\theta)\big)^d.
\ee
Similarly, since the $k$ children with links would need to be blocked
from reaching distance $m-1$,
we also have
\be\begin{split}\label{PA1B}
\PP^{(\theta)}_m(A_1\cap B_m^\rho)&=
\frac{\theta}{Z_m}\sum_{k=0}^d
\binom{d}{k} (e^{-\b})^{d-k} (e^{-\b}\tfrac\b\theta)^k
Z_{m-1}^{d-k}\EE_{m-1}[\theta^\ell \one_{B^\rho_{m-1}}]^k 
\\
&=z_m
\big(e^{-\b/\theta}(1+\zeta_{m-1}\tfrac\b\theta)\big)^d.
\end{split}\ee

For the event $A_2$, we decompose it as 
$A_2=A_2^{\mathrm{mix}}\cup A_2^{\mathrm{same}}$, according as the 2
links from $\rho$ to $x$ are different sorts (crosses/double-bars) or
the same  (Fig \ref{fig two-connection}).  
\bfig
\includegraphics[width=65mm]{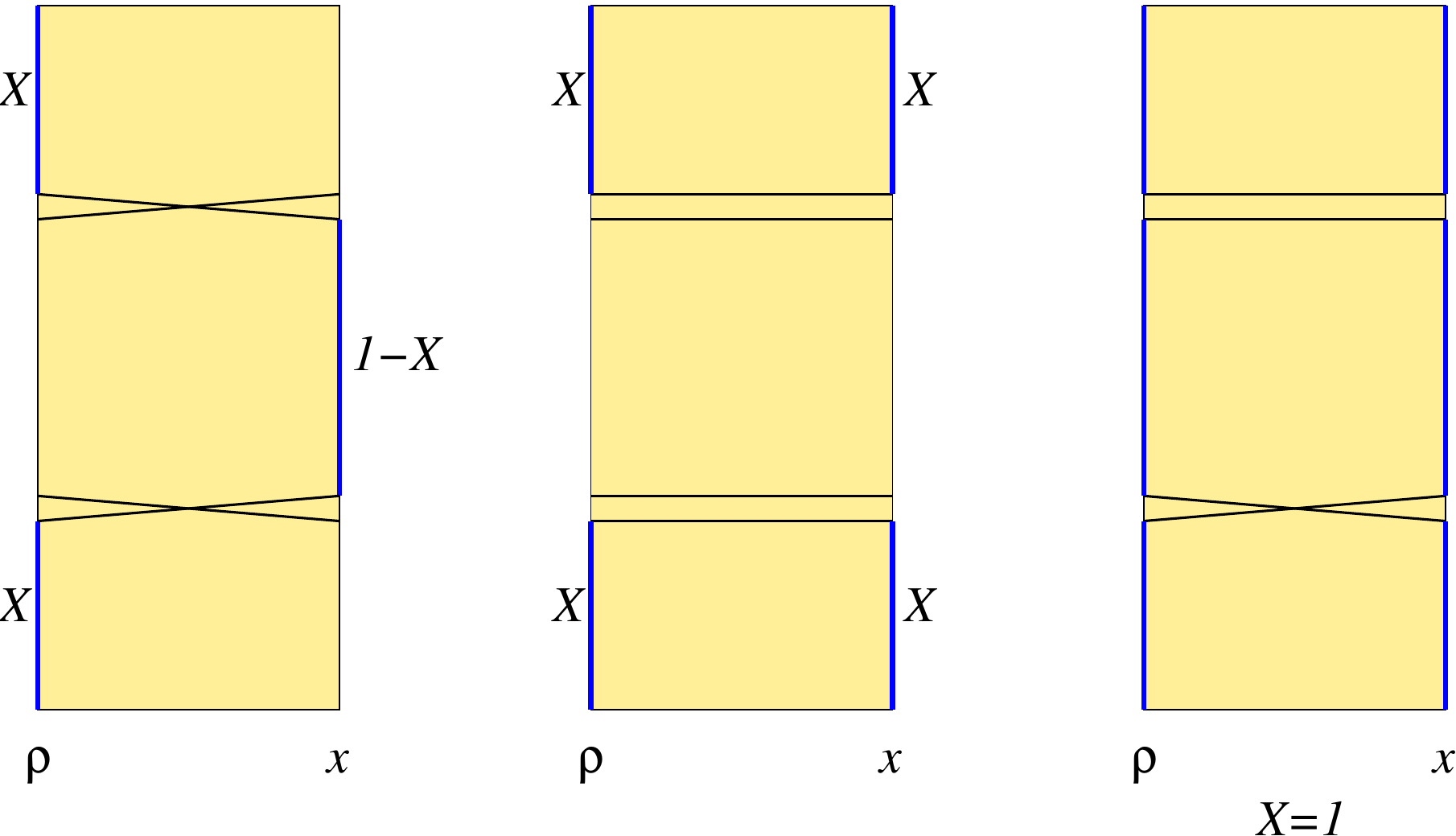}
\caption{Illustration of the possibilities for $\om_{\rho x}$ on the
  event $A_2$.  On $A_2^{\mathrm{same}}$ there are two loops, one of
  which contains $(\rho,0)$;  on $A_2^{\mathrm{mix}}$ only one.  The
  latter is thus more advantageous for long connections.  The random
  variable $X$ has mean $\tfrac23$.}
\label{fig two-connection}
\efig 
If we look at the restriction of $\om$ to the link $\rho x$ only
(i.e., at $\om_{\rho x}$) then it has two loops on $A_2^{\mathrm{same}}$
and a single loop on $A_2^{\mathrm{mix}}$.  Let us number the 
children of $x$ together with the children of $\rho$ excepting $x$ by
$i=1,\dotsc,2d-1$.  Then we have that
\be\label{ell-A2}
\ell=\textstyle\sum_{i=1}^{2d-1} \ell_i -k +
\left\{
\begin{array}{ll}
1 & \mbox{on } A_2^{\mathrm{mix}},\\
2 & \mbox{on } A_2^{\mathrm{same}},
\end{array}
\right.
\ee
where $k$ denotes the total number of 1-links at $\rho$ and at $x$.
To see this one may again imagine that the 1-links are placed last,
one at a time.  If $k=0$ then \eqref{ell-A2} holds due to our
observation about $A_2^{\mathrm{mix}}$ and $A_2^{\mathrm{same}}$
above, if $k>0$ then each link we place merges two previously disjoint
loops.

Let $\L$ denote the loop in $\om_{\rho x}$ containing $(\rho,0)$, and
let $\L_\rho=\L\cap(\{\rho\}\times[0,1])$ and
$\L_x=\L\cap(\{x\}\times[0,1])$ 
denote the parts of $\L$ at $\rho$ and at $x$,
respectively.  
For $B_m^\rho$ to happen, children of $\rho$ which link to $\L$ need
to be blocked from distance $m-1$ and children of $x$ which link to
$\L$ need to be blocked from distance $m-2$;  the remaining children
of $\rho$ and $x$ do not need to be blocked.  
In particular, on $A_2^{\mathrm{mix}}$ all children which link to
either $\rho$ or $x$ need to be blocked.  Write 
$A_2^{\mathrm{mix}}(x,k_0,k_1)$ for the event that (i) $\rho x$
supports one link of each sort, (ii) among the remaining children of
$\rho$ exactly $k_0$ support 1 link and the rest 0, and (iii)
among the children of $x$ exactly $k_1$ support 1 link and the rest 0.
Using \eqref{ell-A2} with $k=k_0+k_1$ and 
a calculation similar to \eqref{PA1B} we get
\be\begin{split}\label{PA2mix}
\EE_m[\theta^\ell\one_{B^\rho_m}
\one_{A_2^\mathrm{mix}}]
&=
\sum_{x\sim\rho}\sum_{k_0=0}^{d-1}\sum_{k_1=0}^d 
\theta^{-k_0-k_1+1} \EE_m\Big[\theta^{\sum_i\ell_i}
\one_{A_2^{\mathrm{mix}}(x,k_0,k_1)}\one_{B^\rho_{m}}\Big]\\
&=\theta Z_{m-1}^{d-1}{Z_{m-2}^d}
\tfrac{d\b^2e^{-\b}}{2}
2u(1-u)
\big(e^{-\b}(1+\tfrac\b\theta\zeta_{m-1})\big)^{d-1}
\big(e^{-\b}(1+\tfrac\b\theta\zeta_{m-2})\big)^{d}.
\end{split}
\ee
For the case of $A_2^{\mathrm{same}}$ we may start with a similar
decomposition, 
\[
\EE_m[\theta^\ell\one_{B^\rho_m}
\one_{A_2^\mathrm{same}}]
=\sum_{x\sim\rho}\sum_{k_0=0}^{d-1}\sum_{k_1=0}^d 
\theta^{-k_0-k_1+2} \EE_m\Big[\theta^{\sum_i\ell_i}
\one_{A_2^{\mathrm{same}}(x,k_0,k_1)}\one_{B^\rho_m}\Big],
\]
where $A_2^{\mathrm{same}}(x,k_0,k_1)$ is defined as 
$A_2^{\mathrm{mix}}(x,k_0,k_1)$ except for requiring the two links
supported by $\rho x$ to be of the same sort instead.
Here we may then further consider the number $j_0\in\{0,\dotsc,k_0\}$
of links with an endpoint in $\L_\rho$ as well as the number
$j_1\in\{0,\dotsc,k_1\}$
of links with an endpoint in $\L_x$.  As mentioned above, these links
need to be blocked, but the remaining do not.  Recalling that the
locations of links are uniform on $[0,1]$ this means that we obtain a
factor $|\L_\rho|$ (respectively $|\L_x|$) for each of these $j_0$
(respectively, $j_1$) links, and hence
\be\begin{split}\label{PA2same}
\EE_m[\theta^\ell\one_{B^\rho_m}
\one_{A_2^\mathrm{same}}]&
=\theta^2 Z_{m-1}^{d-1}{Z_{m-2}^d}
\tfrac{d\b^2e^{-\b}}{2}
(u^2+(1-u)^2) \\
&\quad\EE\big[
\big(e^{-\b}(1+\tfrac\b\theta\zeta_{m-1}|\L_\rho|
+\tfrac\b\theta(1-|\L_\rho|))\big)^{d-1}
\big(e^{-\b}(1+\tfrac\b\theta\zeta_{m-2}|\L_x|
+\tfrac\b\theta(1-|\L_x|))\big)^{d}
\big].
\end{split}
\ee
Here we have simply written $\EE[\cdot]$ for $\EE_m[\cdot\mid A_2]$,
this expectation is over the choice of crosses or double-bars
and over the  lengths $|\L_\rho|$ and $|\L_x|$ only. 

We note here that  the joint expectations
of $|\L_\rho|$ and $|\L_x|$ may be computed
explicitly.  Indeed, as illustrated in Fig.~\ref{fig two-connection},
there is a random variable $X$ such that
 $\L_\rho$ and $\L_x$ have respective lengths 
$X$ and $1-X$ in the case of two crosses;
$X$ and $X$ in the case of two double-bars;  and 
$|\L_\rho|=|\L_x|=X=1$ in the case of a
mixture.  One may check 
\footnote{The conditional distribution of $X$ equals that of
the length of the segment between two uniform independent
points on a circle (with circumference 1) which contains a given point.
} 
that $\EE_m[X\mid A_2^{\mathrm{same}}]=\tfrac23$.

At this point, let us mention
the following asymptotics, which will be useful several times:
if $\s=O(d^{-1})$ and  $x\in\RR$ then we have
\be\label{exp-1}
\big(e^{-\b/\theta}(1+\tfrac\b\theta-\s x \tfrac\b\theta)\big)^d=1-
\tfrac1d(\sfrac12+x\s d)+\tfrac1{d^2}\big(\sfrac13-\a+x\s d-\a x \s d
    +\tfrac12(\sfrac12+x\s d)^2\big)+O(d^{-3}).
\ee
To compute $\PP(A_2)$ we may remove the enforcement of $B^\rho_m$
in \eqref{PA2mix} and \eqref{PA2same} by
setting $\zeta_{m-1}$ and $\zeta_{m-2}$ to 1 and summing the results together,  
giving
\be\begin{split}\label{PA2}
\PP^{(\theta)}_m({A_2})&=
z_m z_{m-1}  
\tfrac{d\b^2e^{-\b/\theta}}{2\theta}
\big(e^{-\b/\theta}(1+\tfrac\b\theta)\big)^{2d-1}
\big(2u(1-u)+\theta(u^2+(1-u)^2)\big)
\\
&=
z_mz_{m-1} 
\tfrac{\theta}{2d}
\big[1-\tfrac1d\big]
\big(2u(1-u)+\theta(u^2+(1-u)^2) + O(d^{-2})\big).
\end{split}
\ee
For the last step we used \eqref{exp-1} to first order, and that
\be\label{2d-factor}
\tfrac{d\b^2e^{-\b/\theta}}{2\theta}=\tfrac{\theta}{2d}+O(d^{-2}).
\ee

\subsection{Stochastic domination}

In some estimates we will want to approximate
the complicated measure
$\PP_m^{(\theta)}(\cdot)$, which involves counting loops, by some simpler
measure.  For this we use \emph{stochastic domination}.  Let us define
$\b^+=(\b\theta)\vee(\b/\theta)$.  Also let us define
$\EE^+_m$ in the same way as $\EE_m$ but with $\b$ replaced by $\b^+$;
thus the links form independent Poisson processes with rate $\b^+$.
We say that an event $A$ is \emph{increasing} if it cannot be
destroyed by adding more links;  examples of increasing events include
$A_1^c$ and $(A_1\cup A_2)^c$ where $A_1$ and $A_2$ are as defined above.
Stochastic domination tells us that
\be\label{std}
A\mbox{ increasing}\quad\Rightarrow\quad
\PP_m^{(\theta)}(A)\leq \PP_m^+(A).
\ee
\begin{proof}[Proof of \eqref{std}]
We apply \cite[Thm.\ 1.1]{georgii-kuneth}.  Note that 
$\PP_m^{(\theta)}\ll\PP^+_m$ and the density 
$f(\om)=\tfrac{d \PP_m^{(\theta)}}{d\PP^+_m}
\propto \theta^{\ell(\om)}(\tfrac{\b}{\b^+})^{|\om|}$ where  $|\om|$ denotes
the number of links.   Let $\om'$ be obtained from $\om$ by adding a
single link.  This link either splits a loop, merges two loops, or
does not change the number of loops, hence
\[
\tfrac{f(\om')}{f(\om)}=\theta^{\ell(\om')-\ell(\om)}\tfrac{\b}{\b^+}
\in\{\tfrac{\b\theta}{\b^+},\tfrac{\b/\theta}{\b^+},\tfrac{\b}{\b^+}\}.
\]
The result follows since all three possible values are $\leq 1$.
\end{proof}
An immediate consequence of \eqref{std}
is that there is some constant $c>0$ such
that 
\be\label{std-1}
\PP_m^{(\theta)}(A_1^c\cap A_2^c)\leq c/d^2
\quad \mbox{for all } m, d\geq1.
\ee
We now deduce
some information about the asymptotic behaviour of
the numbers $z_m=e^{-d\b(1-1/\theta)}\theta Z_{m-1}^d/Z_m$.
We write
\be\label{q-def}
q=q(\theta,u)=\tfrac\theta2\big(
2u(1-u)+\theta(u^2+(1-u)^2)\big)
\ee
\be\label{r-def}
r=r(\theta,u)=2\theta u(1-u)+\tfrac12\theta^2(u^2+\tfrac43(1-u)^2)
\ee
so that $\a_\ast=1+q-r$.

\begin{proposition}\label{pf-prop}
There is 
a constant $C$ and there are 
functions $\eps^{(j)}_m(d)$,  $j\in\{1,2,3\}$, satisfying 
\be\label{eps-bds}
|\eps^{(j)}_m(d)|\leq C/d^2 
\mbox{ for all $m,d\geq 1$},
\ee
such that
\be\label{pfs1}
z_m=1-\tfrac1d(q-\sfrac12)+\eps^{(1)}_m(d)
\ee
and
\be\label{pfs2}
z_m
\big(1+\tfrac1d(q-\sfrac12)+\eps^{(2)}_m(d)\big)
=1-\eps^{(3)}_m(d).
\ee
\end{proposition}

\begin{proof}
Note that \eqref{pfs1} and \eqref{pfs2} are equivalent, 
hence one may proceed by induction on $m$, 
proving \eqref{pfs2} with the induction hypothesis provided by
\eqref{pfs1}.   For the base case $m=1$ one may establish \eqref{pfs1}
directly, splitting into the cases $A_1$, $A_2$ and $(A_1\cup A_2)^c$
to get
\[
Z_1=\theta^{d+1} e^{-d\b(1-1/\theta)}\Big[
\big(e^{-\b/\theta}(1+\tfrac\b\theta)\big)^d+
q(\theta,u) \tfrac{d\b^2e^{-\b/\theta}}{\theta^2}
\big(e^{-\b/\theta}(1+\tfrac\b\theta)\big)^{d-1}
+\eps_1(d)
\Big],
\]
where $0\leq \eps_1(d)\leq e^{d\b(1-1/\theta)}\PP_1(A_1^c\cap A_2^c)$
satisfies \eqref{eps-bds}. 

For $m>1$,
write $\eps^{(3)}_m(d)=\PP^{(\theta)}_m(A_1^c\cap A_2^c)$, this
satisfies \eqref{eps-bds} by \eqref{std-1}.
From the expressions \eqref{PA1} and \eqref{PA2} we have 
\[\begin{split}
1&= \PP^{(\theta)}_m(A_1)+\PP^{(\theta)}_m(A_2)+
\PP^{(\theta)}_m(A_1^c\cap A_2^c)=
\\
&=
z_m\big(e^{-\b/\theta}(1+\tfrac\b\theta)\big)^d+
z_m z_{m-1} 
\tfrac{d\b^2e^{-\b/\theta}}{\theta^2}
\big(e^{-\b/\theta}(1+\tfrac\b\theta)\big)^{2d-1}
q(\theta,u) +\eps_m^{(3)}(d)
\end{split}\]
Hence, using the asymptotics \eqref{exp-1} and \eqref{2d-factor},
\[
1-\eps^{(3)}_m(d)=
z_m \Big[ 1-\tfrac1{2d}+
z_{m-1} \tfrac1d\big(1-\tfrac1d\big) q +\eps^{(4)}(d)
\Big]
\]
for a function $\eps^{(4)}(d)$ not depending on $m$ but otherwise
satisfying the bounds \eqref{eps-bds}.
Using the induction hypothesis we get
\[
1-\eps^{(3)}_m(d)=
z_m \Big[ 1+\tfrac1{d}(q-\tfrac12)
+\eps_m^{(2)}(d) \Big],
\]
where
\[
\eps^{(2)}_m(d) = \eps^{(4)}(d) - \tfrac q{d^2} \big( q + \tfrac12 \bigr) + \tfrac q{d^3} \big( q - \tfrac12 \bigr) + \tfrac qd \bigl( 1 - \tfrac1d \bigr) \eps^{(1)}_{m-1}(d)
\]
is easily seen to satisfy \eqref{eps-bds}.
\end{proof}

\begin{remark}
From the proposition it follows that
\be\label{two-pfs}
z_mz_{m-1} = 1+O(d^{-1})
\ee
where the $O(\cdot)$ is uniform in $m$.
\end{remark}

We now turn to the details of the proof of Prop.~\ref{ineq-lem}.

\subsection{Proof of the lower bound \eqref{large-ineq}}

We have from the definition 
$\s_m=1-\PP^{(\theta)}_m(B_m^\rho)$ that
\be\label{sigma-lb}
\s_m\geq 
\PP^{(\theta)}_m(A_1)
-\PP^{(\theta)}_m(B_m^\rho\cap A_1)
+\PP^{(\theta)}_m(A_2)
-\PP^{(\theta)}_m(B_m^\rho\cap A_2),
\ee
where we have simply bounded the remaining 
difference involving the event $(A_1\cup A_2)^c$
from below by 0.
Consider first the terms involving $A_1$.
From~\eqref{PA1} and~\eqref{PA1B}, bounding
$\s_{m-1}\geq\tilde\s_{m-1}$,  
and using the asymptotics~\eqref{exp-1} as well as
the estimates Prop.~\ref{pf-prop} on $z_m$ we get 
\be\label{PA1-diff}
\begin{split}
&\PP^{(\theta)}_m(A_1)-\PP^{(\theta)}_m(B_m^\rho\cap A_1)\geq
z_m
\big(\tilde\s_{m-1}-\tfrac{\tilde\s_{m-1}}{d}(\sfrac32-\a)-
\tfrac12\tilde\s_{m-1}^2+ O(d^{-3})\big)\\
&\qquad =\big(1-\tfrac1d (q-\sfrac12)\big)
\big(\tilde\s_{m-1}-\tfrac{\tilde\s_{m-1}}{d}(\sfrac32-\a)-
\tfrac12\tilde\s_{m-1}^2\big) + O(d^{-3})\\
&\qquad =
\tilde\s_{m-1}+\tfrac{\tilde\s_{m-1}}{d}(\a-q-1)-
\tfrac12\tilde\s_{m-1}^2+ O(d^{-3}).
\end{split}\ee

Now consider the terms involving $A_2$.
Using that $\zeta_{m-1},\zeta_{m-2}\leq 1-\tilde\s_{m-1}$, as well as
the asymptotics \eqref{exp-1}
to order $d^{-1}$, 
we deduce from \eqref{PA2mix} that
\be\begin{split}\label{PA2mixB}
\PP^{(\theta)}_m({B_m^\rho} \cap
{A_2^\mathrm{mix}})&\leq
z_m z_{m-1} \tfrac{d\b^2 e^{-\b/\theta}}{2\theta} 2u(1-u) 
\big(e^{-\b/\theta}(1+\tfrac\b\theta-\tilde\s_{m-1}\tfrac\b\theta)
\big)^{2d-1}\\
&=z_m z_{m-1} \frac{\theta}{2d}
\Big[\big(1-\tfrac1d\big)2u(1-u)-
\tilde\s_{m-1}4u(1-u)+ O(d^{-2})
\Big]
\end{split}
\ee
and from \eqref{PA2same} that
\be\begin{split}\label{PA2sameB}
\PP^{(\theta)}_m({B_m^\rho} \cap
{A_2^\mathrm{same}})&\leq
z_m z_{m-1} \tfrac{d\b^2 e^{-\b/\theta}}{2}
(u^2+(1-u)^2) \\
&\qquad \EE\Big[
\big(e^{-\b/\theta}(1+\tfrac\b\theta - 
\tilde\s_{m-1}|\L_\rho| \tfrac\b\theta
)\big)^{d-1}
\big(e^{-\b/\theta}(1+\tfrac\b\theta 
-\tilde\s_{m-1}|\L_x|\tfrac\b\theta)\big)^{d}
\Big]\\
&=
z_m z_{m-1}\frac{\theta^2}{2d} (u^2+(1-u)^2)
\Big[\big(1-\tfrac1d\big) - \tilde\s_{m-1} 
\EE\big(|\L_\rho|+|\L_x|\big)+O(d^{-2})
\Big]\\
&= z_m z_{m-1} \frac{\theta^2}{2d}
\Big[\big(1-\tfrac1d\big)(u^2+(1-u)^2)-
\tilde\s_{m-1}(u^2+\tfrac43(1-u)^2)+ O(d^{-2})
\Big].
\end{split}
\ee
Here we used the properties of $|\L_\rho|$ and $|\L_x|$ stated below \eqref{PA2same}
(eee also Fig.~\ref{fig two-connection}).
Using also \eqref{PA2} and \eqref{two-pfs} we get 
\[\begin{split}
&\PP^{(\theta)}_m(A_2)-\PP^{(\theta)}_m(B_m^\rho\cap A_2)\\
&\quad\geq
z_m z_{m-1} \frac{\theta}{2d}
\Bigg\{
2u(1-u)\Big(
\big[1-\tfrac1d\big]-\big[1-\tfrac1d-2\tilde\s_{m-1}\big]
\Big)\\
&\qquad\qquad+\theta (u^2+(1-u)^2)
\Big(
\big[1-\tfrac1d\big]-
\big[1-\tfrac1d-\tilde\s_{m-1}\frac{u^2+\tfrac43(1-u)^2}{u^2+(1-u)^2}\big]
\Big)+O(d^{-2})
\Bigg\}\\
&\quad=r(\theta,u) \tfrac{\tilde\s_{m-1}}{d} + O(d^{-3}),
\end{split}\]
where $r$ is defined in \eqref{r-def}.
Putting this together in~\eqref{sigma-lb} gives
\[
\s_m\geq 
\tilde\s_{m-1}+\tfrac{\tilde\s_{m-1}}{d}
(\a-[1+q-r])-
\tfrac12\tilde\s_{m-1}^2+ O(d^{-3}).
\]
Since $\a_\ast=1+q-r$ this gives~\eqref{large-ineq}.
\qed

\subsection{Proof of the upper bound \eqref{not-large-ineq}}

Write $\S_m^\rho$ for the complement of $B_m^\rho$, so
that $\s_m=\PP^{(\theta)}_m(\S_m^\rho)$.
Clearly 
\[
\s_m=\PP^{(\theta)}_m(A_1\cap \S_m^\rho)+
\PP^{(\theta)}_m(A_2\cap \S_m^\rho)+
\PP^{(\theta)}_m(A_1^c\cap A_2^c \cap \S_m^\rho).
\]
The following will be proved at the end of this section:
\begin{proposition}\label{tree-bound}
For all $d$ large enough there is a constant $C$ such that 
\[
\PP^{(\theta)}_m(A_1^c\cap A_2^c \cap \S_m^\rho)
\leq \frac C{d^2}  (\s_{m-1}\vee\s_{m-2}).
\]
\end{proposition}
\noindent
Before proving this we show how to deduce \eqref{not-large-ineq}.
We have by taking the difference of the expressions~\eqref{PA1}
and~\eqref{PA1B} that 
\be\begin{split}\label{PA1S}
\PP^{(\theta)}_m(A_1\cap \S_m^\rho)&=
z_m
\Big\{\big(e^{-\b/\theta}(1+\tfrac\b\theta)\big)^d-
\big(e^{-\b/\theta}(1+\tfrac\b\theta-\tfrac\b\theta\s_{m-1})\big)^d
\Big\}\\
&=z_m
\big(e^{-\b/\theta}(1+\tfrac\b\theta)\big)^d
\Big\{1-
\Big(1-\frac{\tfrac\b\theta\s_{m-1}}{1+\b/\theta}\Big)^d
\Big\}\\
&\leq z_m
\big(e^{-\b/\theta}(1+\tfrac\b\theta)\big)^d
\tfrac{d\b}{\theta}(1+\tfrac\b\theta)^{-1}\s_{m-1}.
\end{split}\ee
In the last step we used the concavity of the function
$f(x)=1-(1-x)^d$ to bound $f(x)\leq x f'(0)$.

Similarly using~\eqref{PA2mix} and concavity of
$f(x,y)=1-(1-x)^{d-1}(1-y)^d$ (for $d\geq 3$),
\be\begin{split}\label{PA2mixS}
\PP^{(\theta)}_m({\S_m^\rho} \cap
{A_2^\mathrm{mix}})&=
z_m z_{m-1}
\tfrac{d\b^2e^{-\b/\theta}}{2}
\tfrac2\theta u(1-u)
\big(e^{-\b/\theta}(1+\tfrac\b\theta)\big)^{2d-1}
\Big\{
1-\Big(1-\frac{\tfrac\b\theta\s_{m-1}}{1+\tfrac\b\theta}\Big)^{d-1}
    \Big(1-\frac{\tfrac\b\theta\s_{m-2}}{1+\tfrac\b\theta}\Big)^{d}
\Big\}\\
&\leq z_m z_{m-1}
\tfrac{d\b^2e^{-\b/\theta}}{2}
\tfrac2\theta u(1-u)
\big(e^{-\b/\theta}(1+\tfrac\b\theta)\big)^{2d-1}
\tfrac{\b}{\theta}(1+\tfrac\b\theta)^{-1}
\{(d-1)\s_{m-1}+d\s_{m-2}\}\\
&\leq z_m z_{m-1}
\tfrac{d\b^2e^{-\b/\theta}}{2}
\big(e^{-\b/\theta}(1+\tfrac\b\theta)\big)^{2d-1}
\tfrac{d\b}{\theta}(1+\tfrac\b\theta)^{-1}
\tfrac2\theta u(1-u) (\s_{m-1}+\s_{m-2}).
\end{split}\ee
The same argument applied to~\eqref{PA2same} gives
(with the notation $\EE$ used there) 
\be\begin{split}\label{PA2sameS}
\PP^{(\theta)}_m({\S_m^\rho} \cap
{A_2^\mathrm{same}})&=
z_m z_{m-1} \tfrac{d\b^2e^{-\b/\theta}}{2}
 (u^2+(1-u)^2)
\big(e^{-\b/\theta}(1+\tfrac\b\theta)\big)^{2d-1}\\
&\qquad \cdot
\EE\Big[
1-\Big(1-\frac{\tfrac\b\theta\s_{m-1}|\L_\rho|}{1+\tfrac\b\theta}\Big)^{d-1}
    \Big(1-\frac{\tfrac\b\theta\s_{m-2}|\L_x|}{1+\tfrac\b\theta}\Big)^{d}
\Big]\\
&\leq z_m z_{m-1}
\tfrac{d\b^2e^{-\b/\theta}}{2}
(u^2+(1-u)^2)
\big(e^{-\b/\theta}(1+\tfrac\b\theta)\big)^{2d-1}
\tfrac{\b}{\theta}(1+\tfrac\b\theta)^{-1}\\
&\qquad\cdot
\{(d-1)\s_{m-1}\EE|\L_\rho|
+d\s_{m-2}\EE|\L_x|\}\\
&\leq  z_m z_{m-1}
\tfrac{d\b^2e^{-\b/\theta}}{2}
\big(e^{-\b/\theta}(1+\tfrac\b\theta)\big)^{2d-1}
\tfrac{d\b}{\theta}(1+\tfrac\b\theta)^{-1}\\
&\qquad\cdot 
\{\s_{m-1}\tfrac23 (u^2+(1-u)^2)
+\s_{m-2}(\tfrac13 u^2+\tfrac23 (1-u)^2)\}
\end{split}\ee

Using Prop.~\ref{pf-prop} to estimate $z_m$,
the asymptotics \eqref{exp-1}, as well as 
$\tfrac{d\b}{\theta}(1+\tfrac\b\theta)^{-1}=1+\tfrac{\a-1}{d}+O(d^{-2})$
we see that the right-hand side of \eqref{PA1S} satisfies
\[
z_m
\big(e^{-\b/\theta}(1+\tfrac\b\theta)\big)^d
\tfrac{d\b}{\theta}(1+\tfrac\b\theta)^{-1}\s_{m-1}
=\Big(1+\frac{\a-(1+q)}{d}+O(d^{-2})\Big)\s_{m-1},
\]
where $q=q(\theta,u)$ was defined in \eqref{q-def}.
Similarly, using \eqref{2d-factor} and \eqref{two-pfs},
in the right-hand-sides of \eqref{PA2mixS} and \eqref{PA2sameS} 
we have the factors
\[
z_m z_{m-1} 
\tfrac{d\b^2e^{-\b/\theta}}{2}
\big(e^{-\b/\theta}(1+\tfrac\b\theta)\big)^{2d-1}
\tfrac{d\b}{\theta}(1+\tfrac\b\theta)^{-1}
=\tfrac{\theta^2}{2d}(1+O(d^{-1})).
\]
Hence, bounding also $\s_{m-1}$ and $\s_{m-2}$ by their maximum, 
we have that
\[\begin{split}
\s_m&\leq (\s_{m-1}\vee\s_{m-2})\Big[1+\frac{\a-(1+q)}{d}
+ \frac{\theta^2}{2d}\Big( \tfrac4\theta u(1-u)+\tfrac23 (u^2+(1-u)^2)
          +\tfrac13 u^2+\tfrac23 (1-u)^2\Big)+O(d^{-2})\Big]\\
&\qquad\qquad+\PP^{(\theta)}_m(A_1^c\cap A_2^c \cap \S_m^\rho)\\
&= (\s_{m-1}\vee\s_{m-2})\Big[1+\frac{\a-(1+q-r)}{d}+O(d^{-2})
\Big]+\PP^{(\theta)}_m(A_1^c\cap A_2^c \cap \S_m^\rho),
\end{split}\]
where $r=r(\theta,u)$ was defined in \eqref{r-def}.
In the above, all $O(d^{-2})$ terms are uniform in $m$.
Since $1+q-r=\a_\ast$ we see
that \eqref{not-large-ineq} follows once we prove
Prop.~\ref{tree-bound}. 

In the following argument we will examine the subtree 
$\check T$ of $T_m$ which
contains the root and is spanned by edges supporting at least two links.
In $\check T$, the loop-structure is very complicated and we will
not attempt to keep track of it.  Instead we use that $\check T$ is
likely to be small, and that a loop exiting it must do so across an
edge supporting exactly one link, which is a simpler situation to
analyze.  
Roughly speaking, the enforcement of the event 
$A_1^c\cap A_2^c$ will give rise to the factor $d^{-2}$,
and the requirement that the loop exits $\check T$ will give a factor
$\s_{m-k}$ for some $k\geq1$, which can then be bounded in terms of
$\s_{m-1}\vee\s_{m-2}$. 
The details are quite technical.

\begin{proof}[Proof of Prop.~\ref{tree-bound}]
We begin by defining $\check T$ carefully:  we let $\check T$ be the
(random) subtree of $T_m$ containing  
\begin{enumerate}
\item the root $\rho$
\item any vertex in generation 1 with $\geq2$ links to $\rho$,
\item in general, any vertex in generation $k$ with $\geq2$ links to
  some vertex of $\check T$ in generation $k-1$.
\end{enumerate}
Note that $A_1^c\cap A_2^c$ is precisely the event that $\check T$ has
at least two edges.
Let $V_k(\check T)$ denote the set of vertices in $\check T$ in
generation $k$.
For $x$ a vertex of $\check T$, $x\not\in V_m(\check T)$, 
let $d_x$ denote its number of
descendants \emph{not in $\check T$}.  Thus $x$ has $d_x$ outgoing edges
carrying only 0 or 1 links of $\om$.  
For $0\leq k\leq m-1$ we let $\cE_k$ denote the set of outgoing
edges from generation $k$ (to generation $k+1$) which carry 
precisely 1 link.

Note that if the loop of $(\rho,0)$ reaches generation $m$ then either
it reaches generation $m$ within $\check T$, or it passes some link of
$\cup_{k=0}^{m-1} \cE_k$.  Let us by convention set $\s_{-1}=1$ and
$|\cE_{m}|=|V_m(\check T)|$.  We claim that
\be\label{tree-sum}
\PP^{(\theta)}_m(A_1^c\cap A_2^c \cap \S_m^\rho)
\leq \sum_{k=0}^{m} \s_{m-k-1} 
\EE_m^{(\theta)}[|\cE_k| \one_{A_1^c\cap A_2^c}].
\ee
Intuitively, this is because if the loop exits $\check T$ 
through some edge in
$\cE_k$, then it has distance $m-k-1$ left to go to reach the 
$m^{\text{\tiny th}}$
generation of $T_m$.   A detailed justification of \eqref{tree-sum}
requires dealing with the dependencies caused by the factor
$\theta^\ell$.

To do this, let us introduce the following notation.  First, let
$\check \om$ denote the restriction of $\om$ to  $\check T$.  Next,
let $\partial^+\check T$ denote the set of vertices 
$y\in T_m\sm\check T$ whose parent belongs to $\check T$,
and write $\om_y$ for the restriction of $\om$ to the subtree rooted
at $y$.  For simplicity, in the rest of this proof we simply write
$\EE$ for $\EE_m$.  We will make use of the fact that, given
$\check\om$, the random collections $(\cE_j)_{j=0}^{m-1}$ and
$(\om_y)_{y\in\partial^+\check T}$ are conditionally independent
under $\EE$.
This implies that for three functions 
\be\label{3fns}
F_1(\check \om),\quad F_2(\check\om,(\cE_j)_{j=0}^{m-1}),
\quad F_3(\check\om,(\om_y)_{y\in\partial^+\check T})
\ee
we have 
\begin{multline}\label{3fns-e}
\EE\big[F_1(\check \om) F_2(\check\om,(\cE_j)_{j=0}^{m-1})
F_3(\check\om,(\om_y)_{y\in\partial^+\check T})\big]\\
=\EE\big[F_1(\check\om)
\EE[F_2(\check\om,(\cE_j)_{j=0}^{m-1})\mid\check\om]
\EE[F_3(\check\om,(\om_y)_{y\in\partial^+\check T})\mid\check\om]\big].
\end{multline}
Note that we have the decomposition (similar to \eqref{ell-A1})
\[
\ell=\check\ell+\sum_{j=0}^{m-1}\Big[
\sum_{x\in V_j(\check T)} \sum_{i=1}^{d_x} \ell_i^{(x)} 
-|\cE_j|
\Big],
\]
where $\check \ell$ denotes the number of loops in the configuration
 $\check \om$, and $\ell_i^{(x)}$ denotes the number of
loops in the subtree rooted at the $i^{\text{\tiny th}}$ descendant of $x$ not
belonging to $\check T$ (in some numbering of
these descendants).  Hence 
\be\label{th-prod}
\theta^\ell=\theta^{\check\ell} 
\Big(\prod_{j=0}^{m-1} \theta^{-|\cE_j|}\Big)
\Big(
\prod_{j=0}^{m-1} \prod_{x\in V_j(\check T)} \prod_{i=1}^{d_x}
\theta^{\ell_i^{(x)}}\Big)
\ee
is a factorization into three functions as in \eqref{3fns}.
Turning to \eqref{tree-sum}, by considering the possibilities that
either $\check T$ reaches generation $m$ (meaning 
$V_m(\check T)\neq\es$) or that loop of $(\rho,0)$ passes some edge 
$e\in\cup_{k=0}^{m-1}\cE_k$, we have 
\be\label{tree-prod-1}
\PP^{(\theta)}_m(A_1^c\cap A_2^c \cap \S_m^\rho)\leq
\EE^{(\theta)}_m[|V_m(\check T)|\one_{A_1^c\cap A_2^c}]+
\sum_{e} \sum_{k=0}^{m-1}
\PP^{(\theta)}_m(A_1^c\cap A_2^c\cap \{e\in \cE_k\}
\cap \{(e^+,t^+)\lra m\}], 
\ee
where the first sum is over all edges $e$ of $T_m$, and 
$\{(e^+,t^+)\lra m\}$ denotes the event that the further (from $\rho$)
endpoint $(e^+,t^+)$ of the unique link at $e$ lies in a loop of
$\om_{e^+}$ reaching the $m^{\text{\tiny th}}$ generation of $T_m$.   Applying
\eqref{3fns-e} and \eqref{th-prod} we have
\begin{multline*}
\PP^{(\theta)}_m(A_1^c\cap A_2^c\cap \{e\in \cE_k\}
\cap \{(e^+,t^+)\lra m\}]\\=
\frac1{Z_m} \EE\Big[\one_{A_1^c\cap A_2^c} \theta^{\check\ell}
\EE\big[\one\{e\in\cE_k\}\prod_{j=0}^{m-1} \theta^{-|\cE_j|}\mid \check T\big]
\s_{m-k-1}
\prod_{j=0}^{m-1} \prod_{x\in V_j(\check T)} Z_{m-j-1}^{d_x} 
\Big].
\end{multline*}
Taking out the factor $\s_{m-k-1}$, applying \eqref{3fns-e} again in
reverse, and putting back into \eqref{tree-prod-1}, we obtain
\eqref{tree-sum}.  

We proceed by bounding the expectations
\[
\EE_m^{(\theta)}[|\cE_k| \one_{A_1^c\cap A_2^c}]=
\frac1{Z_m} \EE[\theta^\ell|\cE_k| \one_{A_1^c\cap A_2^c}].
\]
Arguing as above we get:
\[
\EE[\theta^\ell|\cE_k| \one_{A_1^c\cap A_2^c}]=
\EE\Big[\one_{A_1^c\cap A_2^c} \theta^{\check\ell}
\prod_{j=0}^{m-1} \prod_{x\in V_j(\check T)} Z_{m-j-1}^{d_x}
\EE\big[|\cE_k|\prod_{j=0}^{m-1} \theta^{-|\cE_j|}\mid \check T\big] 
\Big].
\]
The $|\cE_j|$ are conditionally independent given $\check T$, hence
\[
\EE\big[|\cE_k|\prod_{j=0}^{m-1} \theta^{-|\cE_j|}\mid \check T\big] 
=\EE\big[|\cE_k| \theta^{-|\cE_k|}\mid \check T\big] 
\prod_{j\neq k} \EE\big[\theta^{-|\cE_j|}\mid \check T\big].
\]
Let $p_i=e^{-\b}\b^i/i!$ denote the probabilities of a
Poisson($\beta$) random variable.
Direct computation gives 
\[
\EE\big[\theta^{-|\cE_k|}\mid \check T\big] =
\prod_{x\in V_k(\check T)} 
\Big(\frac{p_0+p_1/\theta}{p_0+p_1}\Big)^{d_x}
\]
and (e.g.\ by differentiating the previous expression)
\[
\EE\big[|\cE_k| \theta^{-|\cE_k|}\mid \check T\big] =
\frac{p_1/\theta}{p_0+p_1/\theta}
\Big(\sum_{x\in V_k(\check T)}d_x\Big)
\prod_{x\in V_k(\check T)} 
\Big(\frac{p_0+p_1/\theta}{p_0+p_1}\Big)^{d-d_x}.
\]
Hence
\[
\frac{\EE\big[|\cE_k| \theta^{-|\cE_k|}\mid \check T\big]}
{\EE\big[\theta^{-|\cE_k|}\mid \check T\big]}\leq
\frac{d p_1/\theta}{p_0+p_1/\theta} |V_k(\check T)|,
\]
and
\[
\EE[\theta^\ell|\cE_k| \one_{A_1^c\cap A_2^c}]\leq 
\frac{d p_1/\theta}{p_0+p_1/\theta} 
\EE\Big[|V_k(\check T)|\one_{A_1^c\cap A_2^c} \theta^{\check\ell}
\prod_{j=0}^{m-1} \prod_{x\in V_j(\check T)} Z_{m-j-1}^{d_x}
\EE\big[\prod_{j=0}^{m-1} \theta^{-|\cE_j|}\mid \check T\big] 
\Big].
\]
Applying \eqref{3fns-e} in reverse it follows that
\[
\EE_m^{(\theta)}[|\cE_k| \one_{A_1^c\cap A_2^c}]\leq
\frac{d p_1/\theta}{p_0+p_1/\theta} 
\EE^{(\theta)}_m\big[|V_k(\check T)|\one_{A_1^c\cap A_2^c} \big].
\]
We bound the last expectation using stochastic domination.  Indeed,
both $|V_k(\check T)|$ and $\one_{A_1^c\cap A_2^c}$ are increasing 
functions of $\om$.   
Hence from \eqref{std}
\[
\EE^{(\theta)}_m\big[|V_k(\check T)|\one_{A_1^c\cap A_2^c} \big]\leq
\EE^+\big[|V_k(\check T)|\one_{A_1^c\cap A_2^c} \big].
\]
Write $p_i^+=e^{-\b^+}(\b^+)^i/i!$ for the
Poisson probabilities with parameter $\beta^+$,
and $p^+_{\geq i}=p^+_i+p^+_{i+1}+\dotsb$.
By a recursive computation using independence we see that 
\[
\EE^+|V_k(\check T)|=(d p_{\geq2}^+)\EE^+|V_{k-1}(\check T)|=
\dotsb=(d p_{\geq2}^+)^k.
\]
We also have 
\[
|V_k(\check T)|\one_{A_1}=\d_{k,0}\one_{A_1},\qquad
|V_k(\check T)|\one_{A_2}=(\d_{k,0}+\d_{k,1})\one_{A_2}.
\]
Using \eqref{tree-sum} we find that
\[\begin{split}
\PP^{(\theta)}_m(A_1^c\cap A_2^c \cap \S_m^\rho)
&\leq \frac{d p_1/\theta}{p_0+p_1/\theta} 
\sum_{k=0}^{m} \s_{m-k-1} 
\EE^+[|V_k(\check T)| \one_{A_1^c\cap A_2^c}]\\
&= \frac{d p_1/\theta}{p_0+p_1/\theta} 
\Big(
\s_{m-1}(1-\PP^+(A_1\cup A_2))+
\s_{m-2} (dp_{\geq2}^+ -\PP^+(A_2))\\
&\qquad\qquad\qquad
+\sum_{k=2}^{m} \s_{m-k-1} (dp_{\geq2}^+)^k
\Big)\\
&\leq c_0\Big(\s_{m-1}\frac{c_1}{d^2}+
\s_{m-2}\frac{c_2}{d^2}+
\sum_{k=2}^{m} \s_{m-k-1}\big(\frac{c_3}{d}\big)^{k}
\Big),
\end{split}\]
for constants $c_0,\dotsc,c_3$ uniform in $d$.
Now we use that there is some $c_4>0$, uniform in $d$,
such that  $\s_{m-1}\leq c_4 \s_m$ for all $m\geq 0$.  
(This can be seen e.g.\ by considering the event that $A_1$ occurs and
that $(x,t_x)$ lies in a loop reaching generation $m$ in its subtree,
where $x$ is some fixed child of $\rho$ and $t_x$ is the `time' of the
incoming link from $\rho$.  This gives $\s_m\geq\s_{m-1}
\PP_m^{(\theta)}(A_1)$.) 
It follows that
\[
\PP^{(\theta)}_m(A_1^c\cap A_2^c \cap \S_m^\rho)\leq
\frac{C'}{d^2}\Big(\s_{m-1}+\s_{m-2}+
\s_{m-2}\sum_{k=2}^{\oo} \big(\frac{c_3c_4}{d}\big)^{k-2}.
\Big).
\]
The last sum converges if $d$ is large enough, and this establishes
Prop \ref{tree-bound}.  
\end{proof}

\smallskip
\noindent
{\bf Acknowledgments:}
We thank the anonymous referees for  several
helpful suggestions to improve the presentation.
The research of JEB is supported by Vetenskapsr{\aa}det grant
2015-05195.
\smallskip

{
\renewcommand{\refname}{\small References}
\bibliographystyle{symposium}

}

\end{document}